\documentclass[11pt]{article}
\usepackage[hmargin=1in,vmargin=1in]{geometry}
\usepackage{comment}
\usepackage{palatino}
\mathcode`l="8000
\begingroup
\makeatletter
\lccode`\~=`\l
\DeclareMathSymbol{\lsb@l}{\mathalpha}{letters}{`l}
\lowercase{\gdef~{\ifnum\the\mathgroup=\m@ne \ell \else \lsb@l \fi}}%
\endgroup

\usepackage{amssymb, amsthm}

\usepackage{todonotes} 

\usepackage{amsmath, amsthm, amscd}
\usepackage{url}
\usepackage{hyperref}
\hypersetup{
  colorlinks   = true, 
  urlcolor     = blue, 
  linkcolor    = blue, 
  citecolor   = red 
}
\usepackage{cleveref}

\usepackage{graphicx}

\newtheorem{theorem}{Theorem}

\newtheorem*{theorem*}{Theorem}

\newtheorem*{proposition*}{Proposition}

\newtheorem*{remark*}{Remark}
\newtheorem{lemma}[theorem]{Lemma}
\newtheorem*{lemma*}{Lemma}

\newtheorem*{problem*}{Problem}

\newtheorem*{claim*}{Claim}

\newtheorem*{conjecture*}{Conjecture}

\newtheorem*{corollary*}{Corollary}
\newtheorem*{example*}{Example}
\newtheorem{observation}{Observation}

\newtheorem*{definitions*}{Обозначения}

\newtheorem*{definition*}{Определение}

\def\geq{\geqslant}
\def\leq{\leqslant}

\def\N{\mathbb{N}}

\def\mR{\mathcal{R}}

\def\mB{\mathcal{B}}

\def\T{\mathcal T}

\begin{document}
\title{Tree covers of size $2$ for the  Euclidean plane}

\author{Artur Bikeev\footnote{Moscow Institute of Physics and Technology, Email: {\tt
bikeev99@mail.ru} }, Andrey Kupavskii\footnote{Moscow Institute of Physics and Technology, Innopolis University, Email: {\tt kupavskii@ya.ru} }, Maxim Turevskii\footnote{Saint Petersburg State University, Email: {\tt
turmax20052005@gmail.com}}}

\maketitle

\begin{abstract}
For a given metric space $(P,\phi)$, a tree cover of stretch $t$ is a collection of trees on $P$ such that edges $(x,y)$ of trees receive length $\phi(x,y)$, and such that for any pair of points $u,v\in P$ there is a tree $T$ in the collection such that the induced graph distance in $T$ between $u$ and $v$ is at most $t\phi(u,v).$ In this paper, we show that, for any set of points $P$ on the Euclidean plane, there is a tree cover consisting of two trees and with stretch $O(1).$ Although the problem in higher dimensions remains elusive, we manage to prove that for a slightly stronger variant of a tree cover problem we must have at least $(d+1)/2$ trees in any constant stretch tree cover in $\mathbb R^d$.
\end{abstract}

\section{Introduction}
Consider a set of $n$ points $P$ in $\mathbb R^2$ with Euclidean norm. Can we use a sparse structure in order to imitate the metric space induced on $P$? In this paper, we address a question of this sort. Let us introduce some definitions. Consider a metric space $(M, \phi)$ and a finite $P\subset M$. 
Take an edge-weighted graph $G = (P, E, w)$, where  every edge $(p,q)\in E$ has weight $w(p,q) = \phi(p,q)$, and the distance $\rho_G(x,y)$ between any two points $x,y\in P$ is the shortest path distance between $x$ and $y$ in $G$. Clearly, $\rho_G(x,y)\ge \phi(x,y)$ for any $x,y\in P$ and $G$. We say that $G$ is a {\it $t$-spanner for $\phi$} if $\rho_G(x,y)\le t \phi(x,y)$ for any  $x,y\in P.$ The parameter $t$ is called the \textit{stretch} of the spanner $G$.

Spanners were introduced in the late 80s in the context of distributed computing \cite{Peleg, Peleg2}. They found applications in efficient broadcast protocols, synchronizing networks and computing global functions, gathering and disseminating data, and routing.

Naturally, there is a trade-off between the number of edges in $G$ and its stretch. Any metric space could be imitated by a weighted complete graph. At the same time, $G$ has to be connected in order to guarantee finite distances between any two points. In general metrics, in order to reach stretch $3$ we may need $\Omega(n^2)$ edges in the worst case, and  stretch $2k-1$ requires $\Omega(n^{1+1/k})$ edges \cite{Alt93, Peleg}.  Spanners for Euclidean spaces are among the most studied. Chew~\cite{Chew86,Chew89} showed that $O(n)$ edges are enough to get constant stretch. That is, a graph with constant average degree is enough in order to get a spanner with constant stretch for point sets on the  Euclidean plane. This result was refined in \cite{Clarkson87,Keil88} as follows: for any $\epsilon>0$ and a point set $P$ on the Euclidean plane, one can get a $(1+\epsilon)$-spanner $G$ with $O(n/\epsilon)$ edges. In~\cite{LS22} it was shown that the dependence on $\epsilon$ is optimal.  For the Euclidean $d$-space, the situation is similar: we have $(1+\epsilon)$-spanners with  $O(n/\epsilon^{d-1})$ edges~\cite{RS91,ADDJS93}, which is also tight~\cite{LS22}. In \cite{Aron08}, the authors constructed, for any $0\le k<n$, planar spanners with stretch $O(n/k+1)$ and $n+k-1$ edges. They also showed that the dependence between $k$ and the stretch is tight up to a constant.

A related notion is that of a tree cover. Trees are the simplest connected graphs, and it is tempting to use them to approximate metrics. Unfortunately, one tree typically cannot give good approximations. It is known (and not difficult to show) that if $G$ is a tree, then its stretch cannot be better than $\Omega(n)$ for certain $n$-element point sets on the plane. A natural question along this path  is: how good an approximation we could get by using several trees $T_1,\ldots, T_\ell$ instead of one? Consider a metric space $(X,\phi)$, a finite set $P\subset X,$ and a collection of $\ell$ spanning trees $\mathcal T:=\{T_1,\ldots, T_\ell\}$ on $P$. Define the metric space on each $T_i$ as in the case of spanners. Next, for any two points $x,y\in P$, define $\rho_{\mathcal T}(x,y) = \min_{i\in[\ell]} \rho_{T_i}(x,y)$. We call such a collection $\mathcal T$ of trees a {\it tree cover}. As in the case of spanners, we are interested in tree covers of small stretch $t$, i.e., so that $\rho_{\mathcal T}(x,y)\le t \phi(x,y)$ with some small $t$.

Small tree covers are of significant algorithmic interest since they allow to reduce different distance-related problems to the case of trees, see \cite{KLMS22,ChangCLMST23,ChangCLMST24}. There is a series of results concerning tree covers for Euclidean space that is in parallel with the results on spanners. One of the first results is the Dumbbell theorem \cite{ADM+95}, which guarantees the existence of a tree cover for a point set in $\mathbb R^d$ of stretch $1+\epsilon$
using $O_d(\epsilon^{-d} \log(1/\epsilon))$ trees. One of the corollaries of this result is an efficient algorithm ($O(n\log n)$ time, $O(n)$ space in fixed dimension and for a fixed $\epsilon$) for constructing Euclidean spanners with several desirable properties such as bounded degree and optimal weight. The exponent of $\epsilon$ was later improved to the optimal value $1-d$ in \cite{chang_optimal_2024}. The Dumbbell theorem was generalized to spaces with doubling dimension $d$ in \cite{BFN22}. In \cite{CCLMST23} the authors got a tree cover with stretch $(1+\epsilon)$ using $O(\epsilon^{-3})$ trees for planar graphs.

In this paper, we study tree covers in the regime when the stretch is constant. Specifically, we answer the following question.
\begin{quote} Given a point set $P\subset \mathbf R^2$, what is the smallest size of a tree cover with  stretch $O(1)$ with respect to the Euclidean metric induced on $P$?  \end{quote}
Previously, it was known that $3$ trees are sufficient (this follows from the shifted quadtree construction of \cite{Chan98}), but $1$ tree has stretch $\Omega(n)$ in the worst case (we present a simple argument in Section~\ref{sec1tree}). In this paper, we resolve this question and show that $2$ trees are sufficient as well.

\begin{theorem}\label{thm1}
    For any set $P\subset \mathbb R^2$ of $n$ points there exists a tree cover $\mathcal T$ of size two   with stretch $O(1)$ with respect to the Euclidean distance on $P.$
\end{theorem}

Note that the assumption on $\mathcal T = \{T_1,T_2\}$ implies that the graph $T_1\cup T_2$ is an $O(1)$-spanner for $P$, but is actually much more restrictive.  We can take $O(1)$ to be equal to $40$ in our argument. This could further be improved to $20$ with an argument adjusting that of Gupta \cite{Gupta01}, and probably beyond, but we omit the technical details for clarity.

Due to a very regular nature and simple description of the provided trees, this construction could be potentially useful for different algorithmic tasks for planar sets, such as routing.

A big question in this direction is what is the minimum size of a tree cover for a point set in $\mathbb R^d$. This problem is very challenging, and we did not manage to improve the trivial lower and upper bounds of $2$ and $d+1$, respectively. However, we managed to progress on a slight variant of this problem: the minimum size of a strong tree cover. Namely, we proved that at least $(d+1)/2$ trees are necessary in order to achieve constant stretch. We give precise definitions and results in Section~\ref{sechigh}. We note that the construction of $d+1$ trees with constant stretch from \cite{Chan98} is actually a strong tree cover. In the other direction, we weakened the notion of a tree cover to a low-distance tree cover and sketch a constructon of a low-distance tree cover of size $3$ in $\mathbb R^3.$

In the next section, we prove Theorem~\ref{thm1}. In Section~\ref{sechigh}, we discuss the problem in higher dimensions.

\section{Proof of Theorem~\ref{thm1}}
Denote $X_n = \{(a,b)\in \mathbb R^2: a,b\in [n]\}$. The crucial step of the proof is the following auxiliary theorem, proved in Section~\ref{sec2}.
\begin{theorem}\label{thm2}
    For any $n\in \mathbb N$, there exist a family of planar trees $\mathcal T=\{T_1,T_2\}$ and a point set $L\subset V(T_1)\cap V(T_2),$ such that:
    \begin{enumerate}
        \item[(i)] for any $x,y\in L$ we have $d_\mathcal T(x,y)\le  5\|x-y\|_2+12.$
        \item[(ii)] For any $z\in X_n$ there exists a point $x\in L$ such that $\|z-x\|_2\le 3.$
    \end{enumerate}
\end{theorem}
\textbf{Remark} Our proof gives $d_\mathcal T(x,y)\le 5\|x-y\|_2+12$, which is $(5+\epsilon)\|x-y\|_2$  for points $x,y\in L$ that are far apart. The multiplicative constant $5$ is the best possible for the construction. \\

Equipped with Theorem~\ref{thm2}, we can prove Theorem~\ref{thm1}. But first we need to introduce some definitions. In the definition of the spanner $G$ and the tree cover $\mathcal T$ above, we assumed that the vertex set of the graphs coincides with $P$. In Theorem~\ref{thm2}, on the other hand, the trees connecting points in $L$ may have other vertices in $X_n.$ Let us adjust the definitions and work with weighted graphs  $G=(Q,E,w)$ with $P\subset Q$. Concretely, assume that we are given a metric space $(X, \phi)$ and a point set $P\subset X$. Consider an edge-weighted graph $G = (Q, E, w)$ with $P\subset Q\subset X$, where  every edge $(p,q)\in E$ has weight $w(p,q) = \phi(p,q)$, and the distance $\rho_G(x,y)$ between any two points $x,y\in P$ is the shortest path distance between $x$ and $y$ in $G$.  We say that $G$ is a {\it Steiner $t$-spanner for $(P,\phi)$} if $\rho_G(x,y)\le t \phi(x,y)$ for any  $x,y\in P.$   The difference with the previous definition is that we allow to use {\it Steiner points,} that is, points from $X$ that do not belong to $Q$.
Similarly, we define a {\it Steiner tree cover} by allowing the trees to use Steiner points.

We are ready to prove Theorem~\ref{thm1}. Take an arbitrary set of points $P\subset \mathbb R^2$ and scale it so that the minimal distance in $P$ is at least $D$. Translate it and choose $n\in \mathbb N$ so that for any $x=(x_1,x_2)\in P$ we have $0\le x_i\le n.$ Apply Theorem~\ref{thm2} and find a point set $L$ and a tree cover $\mathcal T  =\{T_1,T_2\}$ satisfying properties (i), (ii). Connect each point $p$ of $P$ to the closest point $x(p)$ in $L$, breaking ties arbitrarily. The edge $px(p)$ gets weight equal to $\|p-x(p)\|_2.$ We add this edge to both trees $T_1,T_2$. This way, we obtain two new trees $T_1',T_2'$, which vertex sets include $P.$ By property (ii) and the triangle inequality, we have $\|p-x(p)\|_2\le 4\le \frac 4D\|x(p)-x(q)\|_2$ for any $p\ne q\in P$.  Let us check that $\mathcal T' = \{T_1',T_2'\}$ is a Steiner tree cover with stretch $O(1).$ Assume for concreteness that (i) holds in the following form: $d_\mathcal T(x,y)\le C\|x-y\|_2$ with some $C>0.$ Using the triangle inequality and property (i), we have
$$\|p-q\|_2\ge (1-8/D)\|x(p)-x(q)\|_2 \ge \frac{(1-8/D)}{5+12/D} d_{\mathcal T}(x(p),x(q)) \ge $$
$$\ge \frac{(1-8/D)}{5+12/D}\big(d_{\mathcal T'}(p,q)-8\big)\ge \big(\frac{1}{5}-\epsilon\big) d_{\mathcal T'}(p,q)$$
for any $\epsilon$, provided $D$ is large enough.

Gupta \cite{Gupta01} proved a `Steiner point removal' result stating that, given a tree metric space $T$ and a point set $P\subset V(T)$, one can construct a tree $T'$ with $V(T') = P$ such that the weight of any edge in $T'$ is equal to the shortest path distance between these vertices in $T$ and, moreover, the stretch of $T'$ w.r.t. $T$ is at most $8$. Applying it to each of $T_1',T_2'$ and the set of points $P$, we get a new tree cover $\mathcal T^* = \{T_1^*, T_2^*\},$ such that $V(T^*_1) = V(T^*_2) = P$ and, moreover, $d_{\mathcal T^*}(p,q)\le 8d_{\mathcal T'}(p,q)$ for any $p,q \in P.$ (Note that, when applying Gupta's result to, say, $T_1'$, we get the trees in which the weights on each edge $xy$ are the path lengths in $T_1$. Replacing this weight by the actual Euclidean length of the segment $xy$ can only decrease the distances and improve the stretch.)  Combining with the displayed inequality above, we conclude that
$$d_{\mathcal T^*}(p,q)\le (40+\epsilon) C \|p-q\|_2.$$
Since this is true for any $\epsilon>0$ for some tree on our vertex set, the inequality is true with stretch $40$ as well.

The argument of Gupta concerning Steiner point removal could be run explicitly and will give a factor of $4$ instead of $8.$ This will amount to getting stretch $20.$ It is possible to go below that, but getting significantly below that would require a non-trivial amount of effort put into optimization over Euclidean and tree distances.

\subsection{Steiner spanners and Steiner tree covers}
This is a short digression on spanners/tree covers with Steiner points. Introducing Steiner points allows reducing the size of spanners. In particular, the lower bound for the number of edges in a Steiner spanner in $\mathbb R^d$ of stretch $(1+\epsilon)$   is  $\Omega(n \epsilon^{-(d-1)/2}$) \cite{LS22}, and spanners  of such size indeed exist \cite{chang_optimal_2024}, up to a $\log 1/\epsilon$ term. Similarly, a Steiner tree cover in the Euclidean space of stretch $(1+\epsilon)$ must consist of $\Omega(\epsilon^{-(d-1)/2}$) trees \cite{LS22}, and such tree covers exist \cite{chang_optimal_2024}. At the same time, as we have seen, for constant stretch, Steiner points can be avoided \cite{Gupta01} at the expense of multiplicative factor $8$ for the stretch. Interestingly, one can improve this to $2+\epsilon$ by using $f(\epsilon)$ non-Steiner trees (i.e., producing a non-Steiner tree cover for a tree), and $2$ is a sharp threshold: reaching $2$ requires $\Theta(\log n)$ trees, and reaching $2-\epsilon$ requires $\Theta(n)$ trees in the worst case \cite{Workshop23}.

\section{Proof of Theorem~\ref{thm2}}\label{sec2}
This section is split into  three subsections. In the first subsection, we shall provide the key construction. It will be  slightly different from the construction we need in order to prove Theorem~\ref{thm2}. In the second subsection, we shall analyze the metric properties of the construction. In the third section, we shall slightly modify it and give the proof of  Theorem~\ref{thm2}.
\subsection{The construction}\label{sec21}
For fixed $n = 2^m$, $m \in \N$, consider an integer triangular grid $H_n^\Delta = \{(x,y) \in \{0,\ldots, n\}^2:\, x+y \leq n \}.$
In what follows, we shall define in parallel (a) a sequence of triangulations, such that every new one refines the previous one; (b) a sequence of non-intersecting red and blue plane\footnote{A {\it plane graph} is a planar graph together with a drawing on the plane} trees with
 vertices in $H_n^{\Delta}$ such that every new red (blue) tree contains the previous one as a subgraph.
 The resulting red and blue trees  on $H_n^\Delta$  are  shown on Fig. \ref{fig:two_trees_alg_res}. They are non-intersecting and together cover all vertices of $H_n^{\Delta}.$

 \begin{figure}[ht]
\center{\includegraphics[scale=0.5, width=230pt]{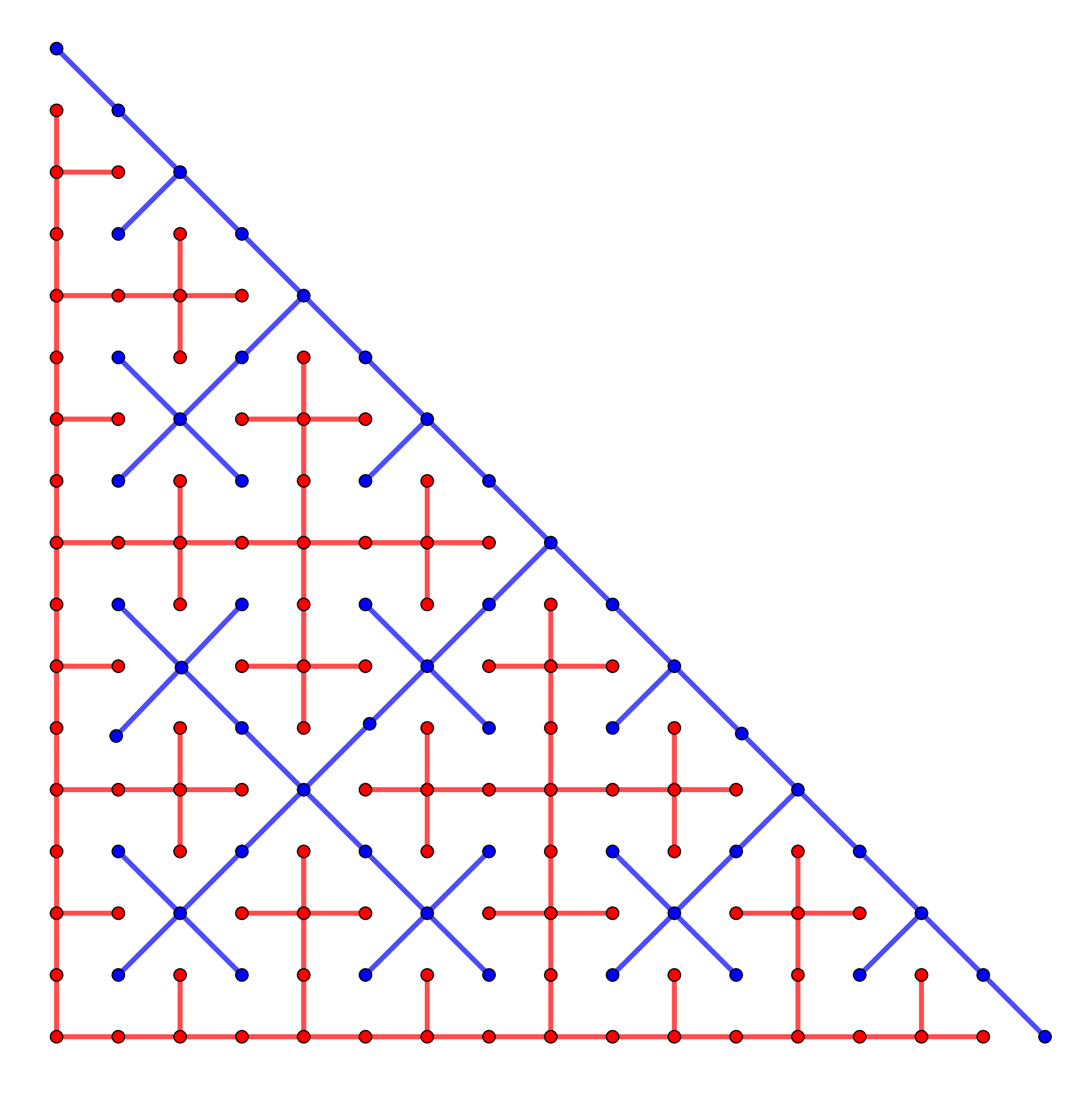}}
\caption{The red and the blue trees $\mR_1, \mB_1$ in $H_n^\Delta$ for $n = 16$}
\label{fig:two_trees_alg_res}
\end{figure}

 \textbf{Step 0 }
 \begin{itemize}
     \item The initial triangulation $\mathcal T_1$ consists of a single triangle  $\tau_{1,1}$ with vertices $(0,0), (0,n), (n,0)$.
     \item The initial blue plane tree $B_1$ is effectively the segment $[(0,n),(n,0)]$, with vertices of $B_1$ being all integer vertices on that segment. More formally, the vertex set is $(i,n-i)$ for $i=0,\ldots, n$, and edges are straight segments connecting points $(i,n-i)$ and $(i+1,n-i-1)$ for $i=0,\ldots, n-1$. In what follows, we tacitly assume that the vertices of the trees that we construct are all integer points on the segments that we add.
     \item The initial red plane tree $R_1$ is the union of segments $[(0,0),(n-1,0)]$ and $[(0,0),(0,n-1)]$.
     \item We say that the catheti of $\tau_{1,1}$ are colored red and the hypotenuse is colored blue. In what follows, we also `induce' the color of the trees that we construct on the edges of triangulations.
     \item The blue segment added to the blue tree at this step gets order $0$. The red segments get order $1$.
 \end{itemize}

\textbf{Step i }
\begin{itemize}
     \item The input to step $i$  is a family $\T_{i} = \{\tau_{i,1},\dots, \tau_{i,2^{(i-1)}} \}$  of right isosceles triangles forming a triangulation of $\tau_{1,1}$, as well as the red plane tree $R_i$ and the blue plane tree $B_i$ that together partition the integer points that lie on the boundary of triangles from $\T_i$.   For odd $i$, the catheti of triangles $\tau_{i,j}$ are colored red, and the hypotenuses are colored blue. For even $i$, the coloring is the way around.
     \item Apply the following procedure to each triangle $\tau \in \T_i$. Let $u$ be the right angle vertex of $\tau$, and let $h=[u,w]$ be the height from $u$ to the hypotenuse of $\tau$. Let $v\ne u$ be the integer node on $h$ that is closest to $u$. Let $\chi$ be the color of the hypotenuse of $\tau$ (note that it depends on the parity of $i$ only). Then add segment $m(\tau):=[v,w]$ to the tree of color $\chi$. See Fig. \ref{fig:trianle steps}. The height $h$ divides the triangle $\tau$ into two smaller triangles $\tau', \tau''$.
     \item Define $\T_{i+1} = \underset{\tau \in \T_i}{\bigcup} \{\tau', \tau'' \}$. 
     \item If, say, the color $\chi$ is red, then put $R_{i+1} = R_i\cup \bigcup_{\tau \in \T_i} \ell(\tau)$, $B_{i+1} = B_i$. If $\chi$ is blue, then the definition is symmetric.
     \item For each $\tau$, $m(\tau)$ gets order $i$. 
 \end{itemize}

 \begin{figure}[ht]
\center{\includegraphics[scale=0.5, width=220pt]{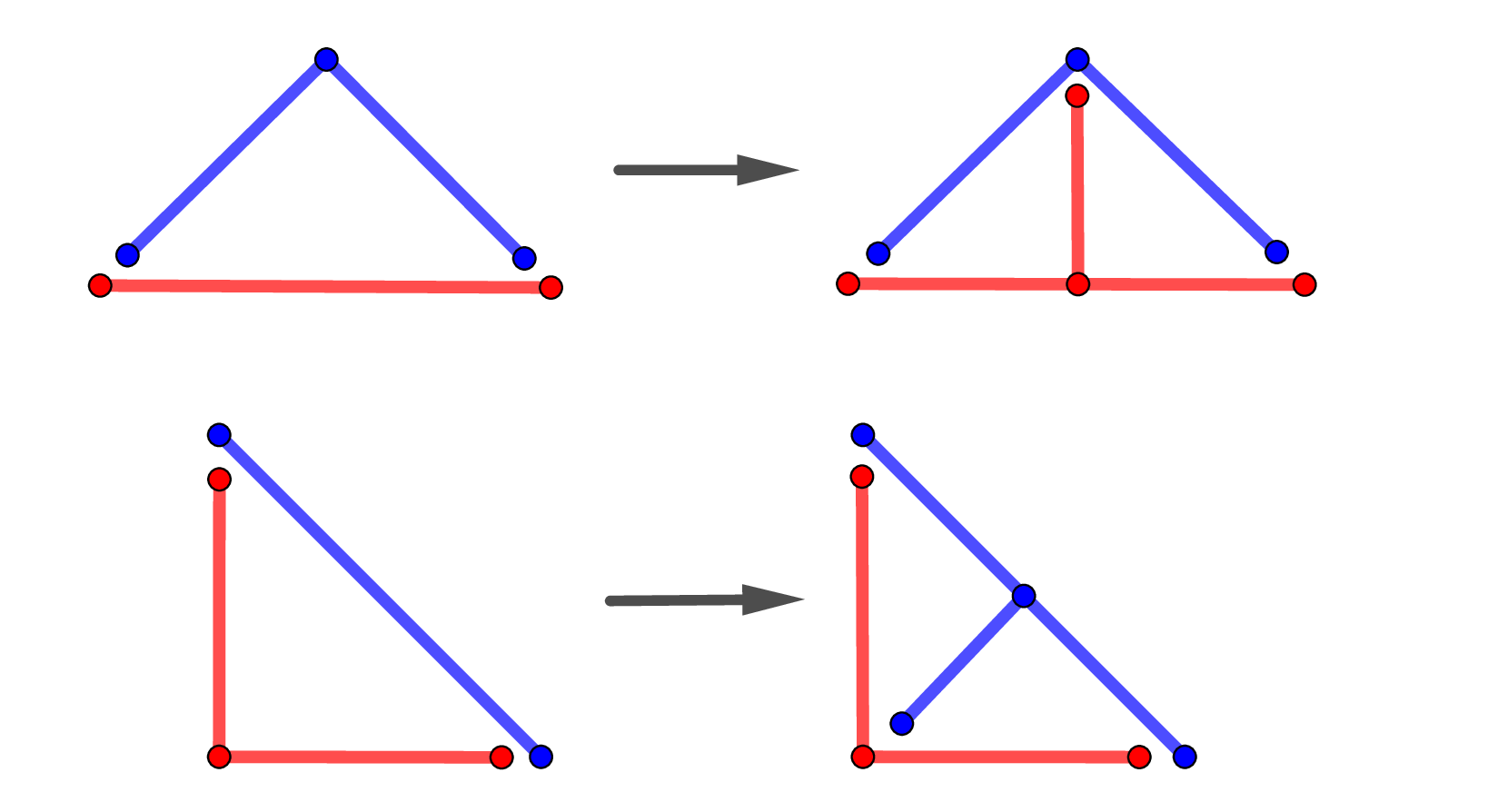}}
\caption{The steps of the algorithm for one triangle $\tau$}
\label{fig:trianle steps}
\end{figure}

 \begin{figure}[ht]
\center{\includegraphics[scale=0.5, width=300pt]{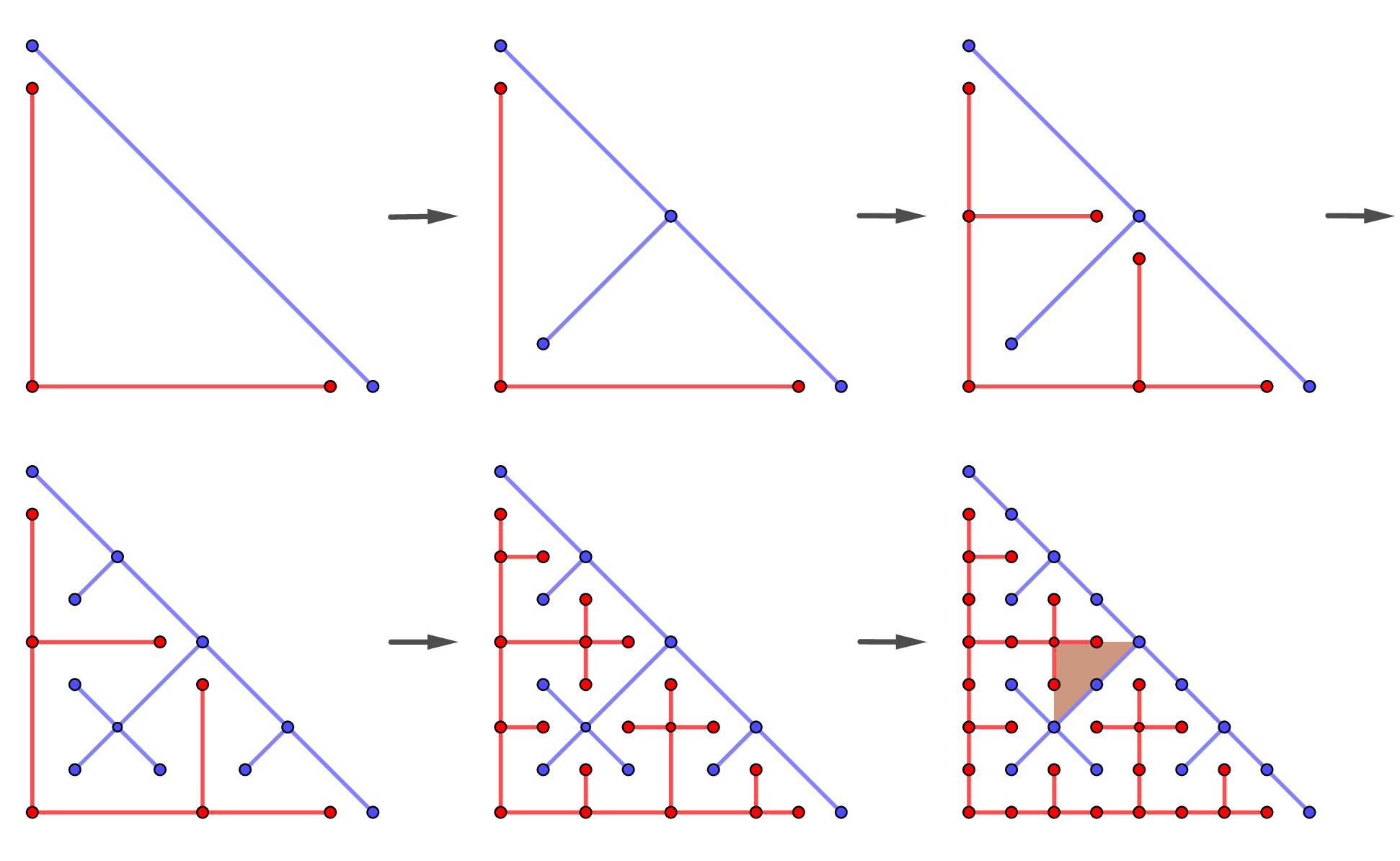}}
\caption{The steps of the algorithm and an example of an atomic triangle}
\label{fig:process}
\end{figure}

 {\bf Termination }

The algorithm stops after $2(m-1)$ steps when there are no integer points in the interior of the triangles of the triangulation. As a result, the set $\T_{2m-1}$ consists of $2^{2(m-1)}$ {\it atomic} triangles with red catheti of Euclidean length $2$ and blue hypotenuses, as is shown on Fig. \ref{fig:process}. Denote by $$\T = \underset{i\in [2(m-1)+1]}{\bigcup} \T_i$$ the set of all triangles obtained in the algorithm. We also get the red and blue trees $R:=R_{2m-1}$, $B:=B_{2m-1}.$ Let us summarize some of the properties of this construction.

\begin{observation}\label{obs}
    The following properties hold.

    (a). Every point of $H_n^\Delta$ is colored, that is, $V(R) \sqcup V(B) = H_n^\Delta$.

    (b)  Fix a triangle $\tau \in \T$. The boundary of $\tau$ contains exactly two uncolored edges, and these edges are contained on two different sides of $\tau$. So, only one of two cases is possible, as is shown on Fig. \ref{fig:noncolored}.a.

    (c). For any red node $x$ there exists a red leaf $\ell$ such that  $\|x-\ell\|_2 \leq \sqrt 5$. For any blue node $y$ there exists a red leaf $\ell$ such that $\|y-\ell\|_2 = 1$.

    (d). For any red leaf $\ell$ there exists a blue node $x(\ell)$ such that $\|\ell-x(\ell)\|_2=1$ and such that, for  every triangle $\tau\in \mathcal T$, if $\ell\in \tau$, then $x(\ell)\in \tau$.
    \end{observation}

    \begin{proof}
    The properties (a), (b) can be easily checked by induction on the step of the algorithm.
    The property (c) holds because any node $x \in H_n^{\Delta}$ is contained in an atomic triangle $\tau$, and $\tau$ has two possible forms, as is shown on Fig. \ref{fig:noncolored}.b. To show property (d), consider an atomic triangle and one of its red leaves $\ell$. In order to get $x(\ell),$ prolong the cathet that $\ell$ lies on until it hits a blue vertex. Every red leaf $\ell$ is formed when a height in some triangle $\tau$ is added to the red tree. Then $x(\ell)$ is simply the right angle vertex of $\tau$. It is easy to see that $x(\ell)$ is contained in all triangles of the triangulation  that $\ell$ is contained in.
    \end{proof}
 \begin{figure}[ht]
\center{\includegraphics[scale=0.5, width=250pt]{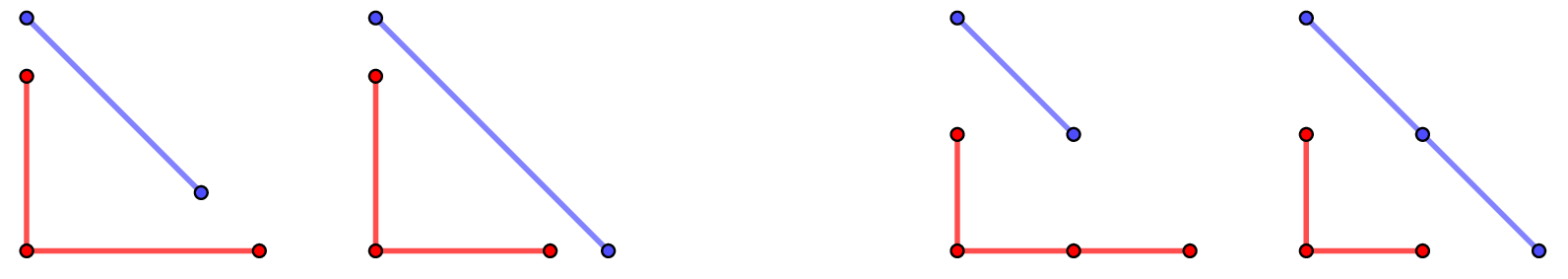}}
\caption{a) Two possible cases of the location of two uncolored edges on the boundary of triangle $\tau_{i,j}$ with red catheti and blue diagonal. b) Two possible cases of the form of the atomic triangles $\tau \in \T_i, i = 2(m-1)+1$.}
\label{fig:noncolored}
\end{figure}

\subsection{Analysis of the distances in the trees}
The edges of the trees $R,B$ naturally get their corresponding Euclidean length. That is, the edges have length $1$ in $R$ and $\sqrt 2$ in $B.$ Then, for any two vertices $x,y$ in $R$ ($B$) let $d_R(x,y)$ ($d_B(x,y)$) stand for the shortest path distance in $R$ ($B$).

The following two lemmas relate the Euclidean and tree distances between interior and boundary points of a triangle $\tau\in \mathcal T$.
\begin{lemma}\label{l:main1}
Let $\tau \in \T_i$ be a triangle obtained in the algorithm, with the hypotenuse of color $\chi\in \{R,B\}.$ Let $x \in \tau$ be an integer point of color $\chi$ and let $u$ be one of the vertices of the triangle $\tau$ that lies on the hypotenuse and is of color $\chi$ (by Observation~\ref{obs} (b), there is at least one such vertex). Then  $d_\chi(x,u)\le 2\sqrt 2\|x-u\|_2$.
\end{lemma}

\begin{proof}
For concreteness, assume that $\chi$ is red  (the proof for the blue color is symmetric). Then the hypothenuse is red. Take the inclusion-minimal triangle $\tau'\subset \tau$ such that, first, $u$ is the vertex (i.e., endpoint) of the hypotenuse $h$ of $\tau'$ and, second, $x\in \tau'$.  Assume that $\tau'$ is the triangle  $ABC$ as on Fig. \ref{fig:triangle abc}, with $A = u$. (The vertex $C$ may be of blue color, but this does not affect the analysis.) Let $f$ denote the Euclidean length of $h$: $f:=\|A-C\|_2$. Then $x$ cannot be located inside the triangle $ADE$ by minimality of $\tau'$, and the closest point to $A=u$ in $ABC\setminus ADE$ w.r.t. Euclidean distance is $E$, which is at distance $\frac 1{2\sqrt 2}f$ from $A.$ At the same time, it is easy to see by induction that $d_R(u,x)\le f$. Indeed, let us traverse $h$ from $u$ towards $x$ and turn at some point $p$. If $p=D$ then we note that $DB$ has the same length as $DC$ and is a hypotenuse of triangles $EDB, FDB$. We then apply induction to one of the triangles $EDB, FDB$ that $x$ belongs to. Otherwise, $x$ is inside the triangle $DFC$, and we can apply induction to that triangle. We get $d_R(u,x)\le f\le 2\sqrt 2 \|A-E\|_2\le 2\sqrt 2 \|u-x\|_2$.
\end{proof}
     \begin{figure}[ht]
\center{\includegraphics[scale=0.5, width=110pt]{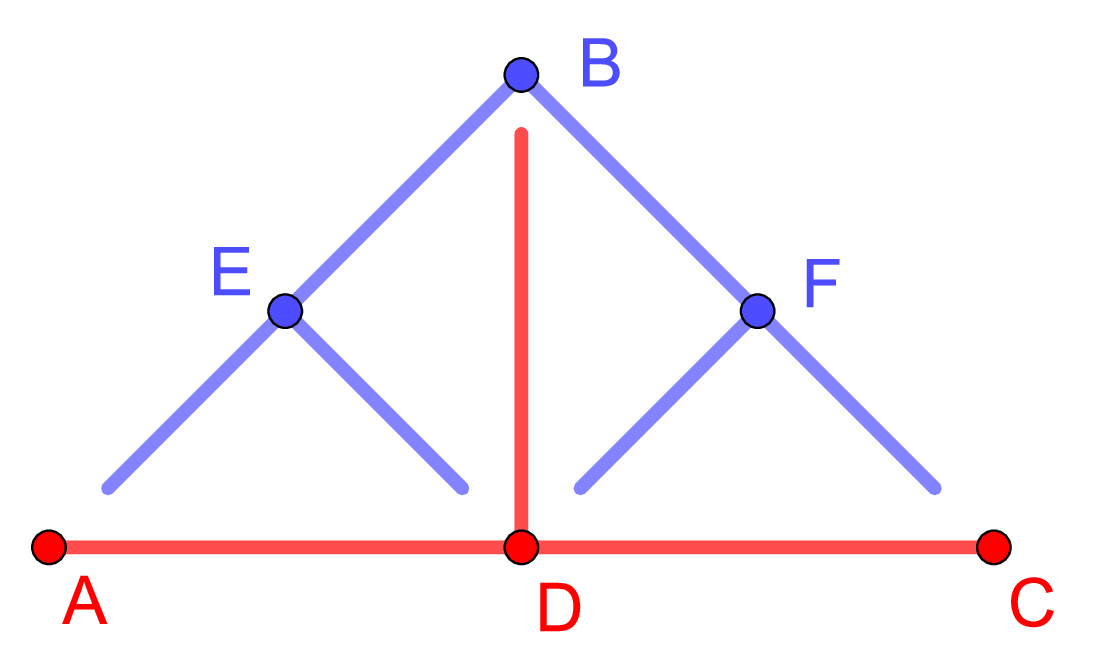}}
\caption{Triangle $\tau'$ with vertices $A,B,C$}
\label{fig:triangle abc}
\end{figure}

\begin{lemma}\label{l:main2}
Let $\tau \in \T_i$ be a triangle obtained in the algorithm, with catheti of color $\chi\in \{R,B\}.$ Let $x \in \tau$ be an integer point of color $\chi$ and let $y$ be an arbitrary (not necessarily integer) point lying on one of the catheti of $\tau$, inside one of the edges of color $\chi$. (That is, $y$ lies on the drawing of the tree of color $\chi$.) Let us include $y$ in the metric space induced by the tree of color $\chi$, by assigning Euclidean distances from it to the endpoints of the edge on which it lies. Then  $d_\chi(x,y)\le 5\|x-y\|_2$.
\end{lemma}

Interestingly, the bound $5$ in the lemma (and thus the multiplicative constant $5$ in the theorem) is tight, as is illustrated on Figure \ref{fig:xyu}.

\begin{figure}[ht]
\center{\includegraphics[scale=0.5, width=100pt]{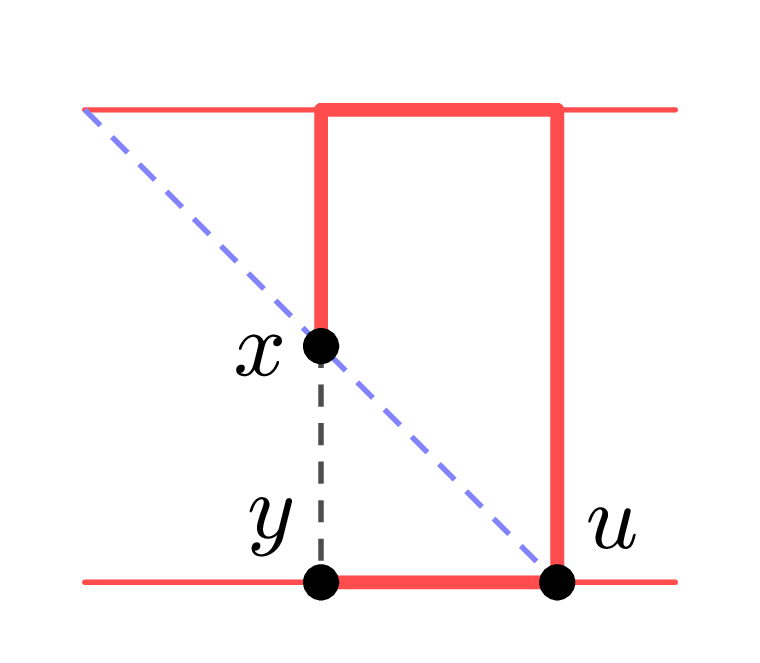}}
\caption{An illustration that the bound $5$ in Lemma~\ref{l:main2}  is tight}
\label{fig:xyu}
\end{figure}

\begin{proof}
   Again, for concreteness, assume that $\chi$ is red. Consider the shortest path from $y$ to $x$ and let $u$ be the first point at which it turns. It is easy to see that the angle $\angle yux$ is at least $\pi/4$. Using this and the law of sines, we see that
   \begin{equation}\label{eqsin} \sqrt 2 \|x-y\|_2 \ge \frac{\|x-y\|_2}{\sin \angle yux} = \frac{\|u-x\|_2}{\sin \angle uyx}=\frac{\|u-y\|_2}{\sin \angle uxy}\ge \max\{\|u-x\|_2,\|u-y\|_2\}.\end{equation}
   We have $\|u-y\|_2 = d_R(u,y)$. At the same time, $u$ is the endpoint of a hypotenuse of a triangle $\tau\in \mathcal T$ that contains $x$. Applying Lemma~\ref{l:main1}, we get that $d_R(u,x)\le 2\sqrt 2 \|u-x\|_2.$ Combining these bounds together, we have
   $$d_R(x,y) = d_R(u,x)+d_R(u,y) \le 2\sqrt 2 \|u-x\|_2+\|u-y\|_2\le 4 \|x-y\|_2+\sqrt 2  \|x-y\|_2.$$
Using basic calculus, one can optimize the function $2\sqrt 2 \|u-x\|_2+\|u-y\|_2$ given the sine conditions \eqref{eqsin} and obtain that the maximum is obtained when $\sin\angle uyx = 1,$
$\sin\angle uxy = \frac 1{\sqrt 2}$ (corresponding to the angles $\pi/2, \pi/4, \pi/4$, the latter being $\angle yux$). As a result, we get $d_R(x,y) \le 5 \|x-y\|_2$.
\end{proof}

\subsection{Proof of Theorem~\ref{thm2}}\label{ssec:33}
Take $n$ as in the statement of the theorem and choose $n'>2n$ such that $n' = 2^m$. For shorthand, we rename $n' = n.$ Construct the trees $B,R$ as in Section~\ref{sec21}
 with $n$. Let $L$ be the set of leaves of $R.$ Note that property (ii) of Theorem~\ref{thm2} is valid for $L$ due to Observation~\ref{obs} (a) and (c). Next, we modify the tree $B$ by including the red leaves into $B$ as follows: for each red leaf $\ell$, connect it to the blue vertex $x$ guaranteed by Observation~\ref{obs} (d). We obtain a tree $B'.$ We take  $R, B'$ as the trees $T_1,T_2$ from the statement of the theorem. We are only left to verify property (i).

Fix any two vertices $\ell_1,\ell_2\in L$. Find a segment $h$ of the smallest order $i$ (cf. Section~\ref{sec21}) that separates $\ell_1$ and $\ell_2$. That is, for some $\tau\in \mathcal T_i$ we have $\ell_1,\ell_2\in \tau$, and then, by drawing the height $h$ in $\tau$, we get two triangles  $\tau_1,\tau_2\in \mathcal T_{i+1}$, such that $\ell_1\in \tau_1$ and $\ell_2\in \tau_2.$  We consider two cases depending on the color of $h$.

If $h$ is red, then let $y$ be the intersection point of $h$ and the segment $\ell_1\ell_2$. Note that $y$ must lie on the red tree, since the only segment on $h$ that does not belong to the red tree is the one closest to the right angle vertex $v$ of $\tau$, and $\ell_1\ell_2$ cannot pass too close to $v$. Triangles $\tau_i$ have red catheti, and so by Lemma~\ref{l:main2} we have $d_R(\ell_j,y)\le 5\|\ell_j-y\|_2$ for $j=1,2$, and thus
$$d_R(\ell_1,\ell_2)\le d_R(\ell_1,y)+d_R(y,\ell_2) \le 5(\|\ell_1-y\|_2+ \|\ell_2-y\|_2)=5\|\ell_1-\ell_2\|_2.$$
(Here, we should note that, although $y$ is not necessarily integral and thus does not belong to $R$, it is easy to see that the first inequality is valid by `rounding' $y$, that is, replacing it with one of the endpoints of the edge it belongs to.)

If $h$ is blue, then the argument is similar. Specifically, we first take the vertices $x(\ell_1),x(\ell_2)$ guaranteed by Observation~\ref{obs} (d). We note that, by the observation, $h$ still separates them.  We apply the argument as above to $x(\ell_1),x(\ell_2)$ and get that
$$d_B(x(\ell_1),x(\ell_2))\le 5\|x(\ell_1)-x(\ell_2)\|_2.$$
Using the triangle inequality several times, we see that
$$d_{B'}(\ell_1,\ell_2) \le 2+d_B(x(\ell_1),x(\ell_2))\le 2+5(2+\|\ell_1-\ell_2\|_2)= 12+5\|\ell_1-\ell_2\|_2.$$
This concludes the proof of the theorem.

\section{One tree is a bad spanner}\label{sec1tree}
In this section, we show that for certain sets of $n$ points on the Euclidean plane  a tree cover consisting of $1$ tree must have stretch $\Omega(n)$.

The point set is a natural one: take $n$ points evenly distributed on a unit circle. Take an arbitrary (non-Steiner) tree spanner $T$ on these $n$ points. Find a centroid $c$ of $T$. Partition the vertices of the graph (excluding $c$) into two groups of vertices of at least $n/3$ points each, that are connected (only) through $c$. Color them red and blue, respectively. Next, since at least a third of the points is red, and at least a third of the points is blue,  none of the colors is fully contained in an arc $A$ of angular length $2\pi/3$ centered at $c.$ Thus, we can find a red-blue pair $r,b$ of consecutive points on the circle that do not belong to $A$. Then
$$d_T(r,b) =d_T(r,c)+d_T(b,c)\ge \|r-c\|_2+\|b-c\|_2.$$
Replacing Euclidean distance with the angular distance, which is the same up to constants independent of $n$, the latter sum is at least $\pi/3+\pi/3 = 2\pi/3.$ At the same time, the actual angular distance between $r,b$ is $2\pi/n.$ This implies  $d_T(r,b) =\Omega(n) \|r-b\|_2$.

\section{Tree covers in higher dimensions}\label{sechigh}
\subsection{Strong tree covers}
What is the minimum size of a tree cover with constant stretch for a point set in $\mathbb R^d$? This value could be anywhere between $2$ and $d+1$. We believe that it should be equal to $d$. Although we did not manage to progress on this question directly, we managed to prove a lower bound for a closely related notion, which we call a strong tree cover. For simplicity, we will work with the integer grid only.

    Let $x$ be a point in $[n]^{d}$. We will use the notation
    $dist(x,y)=\|x_{i}-y_{i}\|_\infty$.
    We define $1-neighbourhood$ of $x$ as a set of points $x'$ such that $dist(x,x') \le 1$.
    For example, the $1-neighbourhood$ of $(1,2)$ in $[3]^{2}$ is $\{(1,1),(1,2),(1,3),(2,1),(2,2),(2,3)\}$.

    Consider a set of trees $\mathcal T = \{T_{1},...,T_{k}\}$ with the vertices being the integer grid $[n]^{d}$. For any $x,y\in [n]^{d}$, denote by $dist_{i}(x,y)$ the distance between $x$ and $y$ in the tree $T_{i}$. We say that $\mathcal T$ is a {\it strong tree cover with stretch $C$} if for any two points $x$ and $y$ in $[n]^{d}$ there exists $1\leq i \leq k$ such that for any point $x'$ from the $1-neighourhood$ of $x$ and for any point $y'$ from the $1-neighourhood$ of $y$ the inequality $dist_{i}(x',y') \leq C \cdot dist(x',y')$ holds. \\
    {\bf Remark: } The difference between this definition and that of a usual tree cover is the neighborhood property: the index $i$ should be the same for all pairs $(x',y')$ in the neighbourhood. In a sense, for each pair of vertices there should be a tree such that they are  somewhat `deep' in that tree.
    \begin{observation}

        If the set of trees $T_{1},...,T_{k}$ is a strong tree cover with stretch $C$ then it is also a tree cover with stretch $C$.
    \end{observation}

Importantly, the construction with quadtrees from \cite{Chan98} gives a strong tree cover for $[n]^{d}$ consisting of $d+1$ trees and with constant stretch. The proof is the same as for the verification of the tree cover property, and we omit it.

The construction from Section \ref{sec2} can be easily modified to a strong tree cover. Recall the trees $R$, $B$, and $B'$ from Section \ref{sec2}.  Let $L_B$ be the set of leaves of $B$. Note that, analogously to Observation \ref{obs}.d, for any blue leaf $l$ of $B$ there exists a red node $y(l)$ such that $||l-x(l)||_2 = 1$ and such that, for every triangle $\tau \in \T$, if $l\in \tau$, then $y(l)\in \tau$. Then we can modify the tree $R$ including the blue leaves as follows: for each blue leaf $l$, connect it to the red vertex $y(l)$. We obtain a tree $R'$ (analogously to $B'$ in Section \ref{ssec:33}).
    \begin{lemma}
        Our example for two trees $R',B'$
        is a strong tree cover.
    \end{lemma}
    \begin{proof}

    Let $u$ and $v$ be two different points in $H_{n}^{\Delta}$. Find a segment $h$ of the smallest order $i$ (analogously to Section 3.3) that $h$ either separates $u$ and $v$ or $u$ and $v$ both are at the distance at most $1$ to this segment. Let $U$ and $V$ be the $1-$neighbourhoods of $u$ and $v$, respectively.

    It is easy to see that, if $u$ and $v$ both are at the distance at most $1$ to this edge, then the distance from each pair $u_2 \in U, v_2 \in V$ is at most $C\cdot ||u_2-v_2||_2$ in the one of the trees $R',B'$ that have the same color as this segment for a constant $C$.

    Let us suppose that $h$ separates $u$ and $v$. Then $h$ is the height of a triangle $\tau \in \T_i$, and $h$ divides this triangle into two triangles $\tau_u, \tau_v \in \T_{i+1}$ such that $u \in \tau_u, v \in \tau_v$. For concreteness, let $h$ be colored red. If each of the vertices $u$ and $v$ is in the interior of $\tau_u$ and $\tau_v$, then, analogously to Lemma \ref{l:main2}, for any $u_2\in U, v_2\in V$ we have $d_{R'}(u_2,v_2) \leq C\cdot ||u_2-v_2||_2$ for a constant $C$. If $u$ in on the blue hypotenuse of $\tau_u$ and $v$ is either in the interior of $\tau_v$ or in the blue hypotenuse of $\tau_v$, then for any $u_2\in U, v_2\in V$ we have $d_{B'}(u_2,v_2) \leq C\cdot ||u_2-v_2||_2$ for a constant $C$. Finally, if $u$ is on the red cathetus of $\tau_u$, we can analogously get that for any $u_2\in U, v_2\in V$ we have $d_{R'}(u_2,v_2) \leq C\cdot ||u_2-v_2||_2$ for a constant $C$.


    \end{proof}

    \begin{theorem}\label{thmlb}
        For any constant $C$ and dimension $d$ there is $n$ such that any strong tree cover of stretch $C$ for the cube $[n]^{d}$ contains at least $k$ trees $T_{1},...,T_{k}$ with $2k \geq d+1$.
    \end{theorem}
    That is, we improve the `trivial' lower bound $2$ to  $d/2$ for the case of strong tree covers. The key ingredient is the connection that we found with Lebesgues's covering dimension.
    \begin{proof}
        Put $B=8(C+1)$. Take a root for each tree in some vertex. Then in each tree remove  all edges such that the depth of the bottom vertex in this edge is divisible by $2(B+1)$. We get some connected components. Each component will have the diameter at most $4B$. Duplicate the original tree and now cut another set of edges: those with the depth of bottom vertex being $B+1$ modulo $2(B+1)$. We again get some connected components of diameter at most $4B$.

        Now, for each tree we have two sets of connected components of diameter at most $4B$ (one after cutting all the edges of depth divisible by $2(B+1)$, and the other after cutting all the edges of the depth $(B+1)$ modulo $2(B+1)$). Note the following: if the distance in the tree between two vertices is at most $B$, then on the path between them either there are no edges with the depth of bottom vertex divisible by $2(B+1)$ or no edges with the depth of bottom vertex $B+1$ modulo $2(B+1)$. So they must lie in one connected component for one of the two cuts.

        Apply the condition of being a strong tree cover for $x=y$. Then we have that for each $x$ there is such $i$, so that in the $1$-neighbourhood of $x$ all vertices are at distance at most $C$ from point $x$ in $T_i$. As $B>8C$,  the depth of $x$ is at least $C+2$ away from either being divisible to $2(B+1)$ or being $B+1$ modulo $2(B+1)$. Thus, all the $1-neighbourhood$ of $x$ are lying in one connected component either for the first or for the second cut for this tree.

        Now, consider the $2k$ sets of components of diameter $\leq 4B$.  For each connected component, we construct a {\it cell}: a closed cube of side length $1$ centered at each vertex of the connected component. We get a connected closed set. Then we consider its interior, getting a connected open set for each connected component. Moreover, sets for different connected components (within one set) are disjoint. Then each point $x$ in the cube $[1,n]^{d}$, now thought of as a body in $\mathcal R^d$ lies inside a cell $C(x)$ with center $y$. For this cell, there is a connected component (in one of the $2k$ sets) that contains $y$ together with its $1$-neighborhood. Therefore, the open set corresponding to that connected component contains $x$ in the interior. Thus, we have an open covering of the cube. Also, the multiplicity (the maximum number of sets containing one point) of this covering is at most $2k$.

        Let us rescale so that the cube becomes a unit cube. We constructed a covering by open sets of unit cube with multiplicity at most $2k$ such that diameter of each set is at most $\epsilon=\frac{4B}{n-1}$. If $2k \leq d$ then it is a contradiction with the fact that Lebesgue covering dimension of a unit cube is $d$.
    \end{proof}
    {\bf Remark: } We actually proved that for a strong cover with  with $k \leq d/2$ trees the stretch is $\Omega(n)$.
\subsection{Low-distance tree coverings}
We managed to show a positive result in $\mathbb R^3$ for a slightly weaker notion of a tree covering (inspired by the proof of Theorem~\ref{thmlb}).
    We will call a set of trees $T_{1},...,T_{k}$ in a cube $B=[n]^{d}$ a \textit{low-distance tree covering of stretch $C$} if for any two points $x$ and $y$ such that $dist(x,y) \leq 1$ there exists such $1\leq i \leq k$ that $dist_{i}(x,y) \leq C$. \\
    \begin{theorem}\label{thm3d}
        There exists a low-distance tree covering for $[n]^3$ with $3$ trees and a constant stretch.
    \end{theorem}
    The proof of this result is technical and uses computer search. We sketch it in the appendix.
\small
\bibliographystyle{alphaurl}
\bibliography{ref}

\section{Appendix. Proof of Theorem~\ref{thm3d}.}

        Since we only need to verify a local property, we will construct forests rather than trees (this is somewhat similar to the proof of Theorem~\ref{thmlb}, where each tree was cut into forests). the forests will possess the necessary weak tree cover property, and then we may connect the connected components of a forest arbitrarily to create a tree.

        In what follows, our goal is to define each forest. In each forest, there is a correspondence between connected components and vertices with coordinates that are all divisible by $4$. We will have three sets of connected components. Each set partitions the vertices in the grid. We do not specify the exact tree structure inside each component since, whichever the structure is, the diameter of each component is bounded. We only need to specify the vertices that belong to each component, and then make sure that the weak tree cover property is satisfied (i.e., that any two points at distance $1$ in the grid belong to the same component in one of the three sets).

        The construction is periodic with  period $8$ along each of the coordinates. The partition between the components will be indicated in the table below. We need to introduce notation for certain vertices. In one block of size $8\times 8\times 8$ we have $8$ vertices with all coordinates divisible by $4$. Namely, for $x,y,z\in\mathbb Z$, denote
        \begin{align*}
            A&:=(8x,8y,8z),\\
            B&:=(8x+4,8y,8z),\\
            C&:=(8x,8y+4,8z),\\
            D&:=(8x+4,8y+4,8z),\\
            E&:=(8x,8y,8z+4),\\
            F&:=(8x+4,8y,8z+4),\\
            G&:=(8x,8y+4,8z+4),\\
            A&:=(8x+4,8y+4,8z+4).
        \end{align*}
        Below, we describe the partition of the vertices in a block $8\times 8\times 8$ into connected components.  For each component set, we have $8\times 8\times 8$ entries in correspondence with the points in the block. The point with coordinates $(8x+a,8y+b,8z+c)$, $0\le a,b,c\le 7$ belongs to the component of the vertex that is written in the $a$-th matrix, $b$-th row and $c$-th column (of the corresponding component set).\vskip+0.1cm
        {\bf Remark 1:} Symbol A (or any other symbol) means the closest vertex of the corresponding form to the point in question. Thus, for example, for $(2,2,6)$ $A$ means $(0,0,8)$, while for $(5,3,2)$ it means  $(8,0,0)$. However, for points at distance $1$ it will always mean the same head. 
   \vskip+0.1cm
    {\bf Remark 2 } Let us comment on how we obtained this construction.
    We have firstly constructed three divisions in $\{0,1,2,3,4\}^{3}$ using SAT-solver, then extended this construction using  symmetry.
    We made the corresponding program code available.
    In order to check the correctness of the construction, first run this (generator) \url{https://pastebin.com/ZMbg38Zt}
    Then copy the output of this strictly after 34918469 and run this (checker) \url{https://pastebin.com/ttp4BrhB} 
    \\
    \vskip+0.1cm
    Our construction is given below: In the $a$-th table in the $b$-th row in the $c$-th column there are three symbols the symbol in the $(8x+a,8y+b,8z+c)$ in the first division, in the second division and in the third division we construct.\\

$a=0$:\\
\begin{tabular}{ | c | c | c | c | c | c | c | c | }
\hline
 (A,A,A) & (A,E,E) & (A,E,E) & (A,E,E) & (E,E,E) & (A,E,E) & (A,E,E) & (A,E,E)  \\
\hline
 (A,C,C) & (A,G,G) & (G,E,G) & (G,E,G) & (G,E,G) & (G,E,G) & (G,E,G) & (A,G,G)  \\
\hline
 (A,C,C) & (G,C,G) & (G,E,G) & (G,E,G) & (G,E,G) & (G,E,G) & (G,E,G) & (G,C,G)  \\
\hline
 (A,C,C) & (G,C,G) & (G,C,G) & (G,G,G) & (G,G,G) & (G,G,G) & (G,C,G) & (G,C,G)  \\
\hline
 (C,C,C) & (G,C,G) & (G,C,G) & (G,C,C) & (G,G,G) & (G,C,C) & (G,C,G) & (G,C,G)  \\
\hline
 (A,C,C) & (G,C,G) & (G,C,G) & (G,G,G) & (G,G,G) & (G,G,G) & (G,C,G) & (G,C,G)  \\
\hline
 (A,C,C) & (G,C,G) & (G,E,G) & (G,E,G) & (G,E,G) & (G,E,G) & (G,E,G) & (G,C,G)  \\
\hline
 (A,C,C) & (A,G,G) & (G,E,G) & (G,E,G) & (G,E,G) & (G,E,G) & (G,E,G) & (A,G,G)  \\
\hline
\end{tabular}
\\
$a=1$:\\
\begin{tabular}{ | c | c | c | c | c | c | c | c | }
\hline
 (A,B,B) & (A,E,F) & (A,E,F) & (A,E,F) & (F,E,F) & (A,E,F) & (A,E,F) & (A,E,F)  \\
\hline
 (A,C,C) & (A,D,G) & (A,D,G) & (A,E,F) & (H,E,F) & (A,E,F) & (A,D,G) & (A,D,G)  \\
\hline
 (A,C,C) & (A,D,G) & (A,D,G) & (G,D,F) & (G,H,F) & (G,D,F) & (A,D,G) & (A,D,G)  \\
\hline
 (A,C,C) & (G,D,C) & (G,D,D) & (G,D,F) & (G,H,F) & (G,D,F) & (G,D,D) & (G,D,C)  \\
\hline
 (D,C,C) & (G,C,H) & (G,H,H) & (G,H,G) & (G,H,G) & (G,H,G) & (G,H,H) & (G,C,H)  \\
\hline
 (A,C,C) & (G,D,C) & (G,D,D) & (G,D,F) & (G,H,F) & (G,D,F) & (G,D,D) & (G,D,C)  \\
\hline
 (A,C,C) & (A,D,G) & (A,D,G) & (G,D,F) & (G,H,F) & (G,D,F) & (A,D,G) & (A,D,G)  \\
\hline
 (A,C,C) & (A,D,G) & (A,D,G) & (A,E,F) & (H,E,F) & (A,E,F) & (A,D,G) & (A,D,G)  \\
\hline
\end{tabular}
\\
$a=2$:\\
\begin{tabular}{ | c | c | c | c | c | c | c | c | }
\hline
 (A,B,B) & (A,B,F) & (A,B,F) & (A,B,F) & (F,F,F) & (A,B,F) & (A,B,F) & (A,B,F)  \\
\hline
 (A,D,D) & (A,D,D) & (B,D,F) & (B,D,F) & (H,H,F) & (B,D,F) & (B,D,F) & (A,D,D)  \\
\hline
 (A,D,D) & (B,D,C) & (B,D,F) & (B,D,F) & (H,G,F) & (B,D,F) & (B,D,F) & (B,D,C)  \\
\hline
 (A,C,C) & (G,D,C) & (G,D,H) & (G,D,F) & (H,H,F) & (G,D,F) & (G,D,H) & (G,D,C)  \\
\hline
 (D,C,C) & (G,H,C) & (G,G,H) & (G,H,H) & (G,H,H) & (G,H,H) & (G,G,H) & (G,H,C)  \\
\hline
 (A,C,C) & (G,D,C) & (G,D,H) & (G,D,F) & (H,H,F) & (G,D,F) & (G,D,H) & (G,D,C)  \\
\hline
 (A,D,D) & (B,D,C) & (B,D,F) & (B,D,F) & (H,G,F) & (B,D,F) & (B,D,F) & (B,D,C)  \\
\hline
 (A,D,D) & (A,D,D) & (B,D,F) & (B,D,F) & (H,H,F) & (B,D,F) & (B,D,F) & (A,D,D)  \\
\hline
\end{tabular}
\\
$a=3$:\\
\begin{tabular}{ | c | c | c | c | c | c | c | c | }
\hline
 (A,B,B) & (A,B,F) & (A,B,F) & (E,B,F) & (F,F,F) & (E,B,F) & (A,B,F) & (A,B,F)  \\
\hline
 (B,B,D) & (B,B,D) & (B,B,D) & (D,B,F) & (H,H,F) & (D,B,F) & (B,B,D) & (B,B,D)  \\
\hline
 (A,D,D) & (A,D,D) & (D,D,D) & (B,D,F) & (H,H,F) & (B,D,F) & (D,D,D) & (A,D,D)  \\
\hline
 (D,D,C) & (G,D,C) & (G,D,H) & (H,D,H) & (H,H,F) & (H,D,H) & (G,D,H) & (G,D,C)  \\
\hline
 (D,D,C) & (G,D,C) & (G,H,H) & (G,H,H) & (G,H,H) & (G,H,H) & (G,H,H) & (G,D,C)  \\
\hline
 (D,D,C) & (G,D,C) & (G,D,H) & (H,D,H) & (H,H,F) & (H,D,H) & (G,D,H) & (G,D,C)  \\
\hline
 (A,D,D) & (A,D,D) & (D,D,D) & (B,D,F) & (H,H,F) & (B,D,F) & (D,D,D) & (A,D,D)  \\
\hline
 (B,B,D) & (B,B,D) & (B,B,D) & (D,B,F) & (H,H,F) & (D,B,F) & (B,B,D) & (B,B,D)  \\
\hline
\end{tabular}
\\
$a=4$:\\
\begin{tabular}{ | c | c | c | c | c | c | c | c | }
\hline
 (B,B,B) & (F,B,F) & (F,B,F) & (B,F,F) & (F,F,F) & (B,F,F) & (F,B,F) & (F,B,F)  \\
\hline
 (D,B,D) & (H,B,D) & (B,B,D) & (B,D,F) & (H,H,F) & (B,D,F) & (B,B,D) & (H,B,D)  \\
\hline
 (D,D,D) & (H,D,D) & (D,D,D) & (B,D,F) & (H,H,F) & (B,D,F) & (D,D,D) & (H,D,D)  \\
\hline
 (D,D,D) & (H,D,H) & (H,D,H) & (H,D,H) & (H,H,F) & (H,D,H) & (H,D,H) & (H,D,H)  \\
\hline
 (D,D,D) & (H,D,H) & (H,D,H) & (H,H,H) & (H,H,H) & (H,H,H) & (H,D,H) & (H,D,H)  \\
\hline
 (D,D,D) & (H,D,H) & (H,D,H) & (H,D,H) & (H,H,F) & (H,D,H) & (H,D,H) & (H,D,H)  \\
\hline
 (D,D,D) & (H,D,D) & (D,D,D) & (B,D,F) & (H,H,F) & (B,D,F) & (D,D,D) & (H,D,D)  \\
\hline
 (D,B,D) & (H,B,D) & (B,B,D) & (B,D,F) & (H,H,F) & (B,D,F) & (B,B,D) & (H,B,D)  \\
\hline
\end{tabular}
\\
$a=5$:\\
\begin{tabular}{ | c | c | c | c | c | c | c | c | }
\hline
 (A,B,B) & (A,B,F) & (A,B,F) & (E,B,F) & (F,F,F) & (E,B,F) & (A,B,F) & (A,B,F)  \\
\hline
 (B,B,D) & (B,B,D) & (B,B,D) & (D,B,F) & (H,H,F) & (D,B,F) & (B,B,D) & (B,B,D)  \\
\hline
 (A,D,D) & (A,D,D) & (D,D,D) & (B,D,F) & (H,H,F) & (B,D,F) & (D,D,D) & (A,D,D)  \\
\hline
 (D,D,C) & (G,D,C) & (G,D,H) & (H,D,H) & (H,H,F) & (H,D,H) & (G,D,H) & (G,D,C)  \\
\hline
 (D,D,C) & (G,D,C) & (G,H,H) & (G,H,H) & (G,H,H) & (G,H,H) & (G,H,H) & (G,D,C)  \\
\hline
 (D,D,C) & (G,D,C) & (G,D,H) & (H,D,H) & (H,H,F) & (H,D,H) & (G,D,H) & (G,D,C)  \\
\hline
 (A,D,D) & (A,D,D) & (D,D,D) & (B,D,F) & (H,H,F) & (B,D,F) & (D,D,D) & (A,D,D)  \\
\hline
 (B,B,D) & (B,B,D) & (B,B,D) & (D,B,F) & (H,H,F) & (D,B,F) & (B,B,D) & (B,B,D)  \\
\hline
\end{tabular}
\\
$a=6$:\\
\begin{tabular}{ | c | c | c | c | c | c | c | c | }
\hline
 (A,B,B) & (A,B,F) & (A,B,F) & (A,B,F) & (F,F,F) & (A,B,F) & (A,B,F) & (A,B,F)  \\
\hline
 (A,D,D) & (A,D,D) & (B,D,F) & (B,D,F) & (H,H,F) & (B,D,F) & (B,D,F) & (A,D,D)  \\
\hline
 (A,D,D) & (B,D,C) & (B,D,F) & (B,D,F) & (H,G,F) & (B,D,F) & (B,D,F) & (B,D,C)  \\
\hline
 (A,C,C) & (G,D,C) & (G,D,H) & (G,D,F) & (H,H,F) & (G,D,F) & (G,D,H) & (G,D,C)  \\
\hline
 (D,C,C) & (G,H,C) & (G,G,H) & (G,H,H) & (G,H,H) & (G,H,H) & (G,G,H) & (G,H,C)  \\
\hline
 (A,C,C) & (G,D,C) & (G,D,H) & (G,D,F) & (H,H,F) & (G,D,F) & (G,D,H) & (G,D,C)  \\
\hline
 (A,D,D) & (B,D,C) & (B,D,F) & (B,D,F) & (H,G,F) & (B,D,F) & (B,D,F) & (B,D,C)  \\
\hline
 (A,D,D) & (A,D,D) & (B,D,F) & (B,D,F) & (H,H,F) & (B,D,F) & (B,D,F) & (A,D,D)  \\
\hline
\end{tabular}
\\
$a=7$:\\
\begin{tabular}{ | c | c | c | c | c | c | c | c | }
\hline
 (A,B,B) & (A,E,F) & (A,E,F) & (A,E,F) & (F,E,F) & (A,E,F) & (A,E,F) & (A,E,F)  \\
\hline
 (A,C,C) & (A,D,G) & (A,D,G) & (A,E,F) & (H,E,F) & (A,E,F) & (A,D,G) & (A,D,G)  \\
\hline
 (A,C,C) & (A,D,G) & (A,D,G) & (G,D,F) & (G,H,F) & (G,D,F) & (A,D,G) & (A,D,G)  \\
\hline
 (A,C,C) & (G,D,C) & (G,D,D) & (G,D,F) & (G,H,F) & (G,D,F) & (G,D,D) & (G,D,C)  \\
\hline
 (D,C,C) & (G,C,H) & (G,H,H) & (G,H,G) & (G,H,G) & (G,H,G) & (G,H,H) & (G,C,H)  \\
\hline
 (A,C,C) & (G,D,C) & (G,D,D) & (G,D,F) & (G,H,F) & (G,D,F) & (G,D,D) & (G,D,C)  \\
\hline
 (A,C,C) & (A,D,G) & (A,D,G) & (G,D,F) & (G,H,F) & (G,D,F) & (A,D,G) & (A,D,G)  \\
\hline
 (A,C,C) & (A,D,G) & (A,D,G) & (A,E,F) & (H,E,F) & (A,E,F) & (A,D,G) & (A,D,G)  \\
\hline
\end{tabular}
\\ \\ \\
    Concluding, we have divided everything into blocks in three groups, such that the diameter of each block is at most $18$. We then construct the corresponding trees, connecting every vertex with its head in each division (distance $\leq 9$), and then connecting the heads from one set into a tree arbitrarily.
    This way, we get $3$ trees, where the distance between every two vertices at distance $\le 1$ is at most $18$. This is a weak tree covering by $3$ trees in $[n]^{3}$.

{\bf Remark 3.}  We may blow up each cell into a cube $K \times K \times K$ and get all vertices at  distance $\leq K$ being at distance $\le f(K)$ for some function of $K$. We yet do not understand, whether we can scale this construction to the usual tree covering, not just the low-distance tree covering.

\end{document}